%% file: main.tex
\documentclass[letterpaper, 10 pt, conference]{ieeeconf}
\IEEEoverridecommandlockouts                              

\input{Preamble.tex}

\begin{document}
\title{\LARGE \bf
A Traffic Management Framework for On-Demand Urban Air Mobility Systems}
\author{Milad~Pooladsanj$^{1}$,
        Ketan~Savla$^{2}$,
        and Petros A. Ioannou$^{1}$
\thanks{$^{1}$M. Pooladsanj and P. Ioannou are with the Department of Electrical Engineering, University of Southern California, Los Angeles, CA 90007 USA {\tt\small (email: pooladsa@usc.edu; ioannou@usc.edu)}.}
\thanks{$^{2}$K. Savla is with the Department of Civil and Environmental Engineering, University of Southern California, Los Angeles CA 90089 USA {\tt\small (email: ksavla@usc.edu)}. K. Savla has financial interest in Xtelligent, Inc.}
\thanks{
This research was supported in part by the PWICE institute at USC.
}
}\maketitle
\begin{abstract}
Urban Air Mobility (UAM) offers a solution to current traffic congestion by providing on-demand air mobility in urban areas. Effective traffic management is crucial for efficient operation of UAM systems, especially for high-demand scenarios. In this paper, we present a centralized traffic management framework for on-demand UAM systems. Specifically, we provide a scheduling policy, called VertiSync, which schedules the aircraft for either servicing trip requests or rebalancing in the system subject to aircraft safety margins and energy requirements. We characterize the system-level throughput of VertiSync, which determines the demand threshold at which passenger waiting times transition from being stabilized to being increasing over time. We show that the proposed policy is able to maximize throughput for sufficiently large fleet sizes. We demonstrate the performance of VertiSync through a case study for the city of Los Angeles, and show that it significantly reduces passenger waiting times compared to a first-come first-serve scheduling policy. 
\end{abstract}

\input{introduction}

\input{problem-formulation.tex}

\input{scheduling}

\input{simulation}

\input{conclusion}

\bibliographystyle{ieeetr}
\bibliography{References_main}

\input{appendix}

\end{document}

%% file: Preamble.tex
\usepackage[utf8]{inputenc}
\usepackage[english]{babel}
\usepackage{cite}
\usepackage{graphicx}
 \graphicspath{{./figures/}}
\usepackage[justification=centering]{caption}
\usepackage{subcaption}
\usepackage{mathtools}
\usepackage{amsmath}

\allowdisplaybreaks
\usepackage{amsthm,amsfonts,amstext,amssymb}
\usepackage{dsfont}
\let\labelindent\relax
\usepackage{enumitem}   
\usepackage{indentfirst}
\usepackage[font=small,skip=0pt]{caption}
\usepackage{float}
\usepackage{xcolor}
\usepackage{xr}
\usepackage{tikz}
\usetikzlibrary{shapes,arrows,positioning,calc}
\usepackage{standalone}
\usepackage[pdftex]{hyperref}
\usepackage[normalem]{ulem}
\usepackage{algorithm}
\usepackage{algpseudocode}

\newtheorem{theorem}{Theorem}

\theoremstyle{definition}
\newtheorem{definition}{Definition}
\theoremstyle{remark}
\newtheorem{remark}{Remark}
\newtheorem{assumption}{Assumption}
\newtheorem{example}{Example}

\algrenewcommand\algorithmicrequire{\textbf{Input:}}
\newcommand{\mpcommentout}[1]{}

\newcommand{\ra}{\rightarrow}

\newcommand{\Z}{\mathbb{Z}}
\newcommand{\N}{\mathbb{N}}
\newcommand{\E}[1]{\mathbb{E}\left[ #1 \right]}

\newcommand{\takeofftau}{\tau}

\newcommand{\cruisetau}{\tau_c}

\newcommand{\vertset}{\mathcal{V}}
\newcommand{\vertcard}[1]{V_{#1}}
\newcommand{\routeset}{\mathcal{P}}
\newcommand{\routecard}{P}
\newcommand{\aircraftset}{\mathcal{A}}
\newcommand{\aircraftcard}{A}
\newcommand{\slotset}[1]{\mathcal{S}_{#1}}
\newcommand{\slotcard}[1]{S_{#1}}

\newcommand{\routingvec}[2]{R_{#1}^{#2}}
\newcommand{\numsafesch}{R}
\newcommand{\safeschset}{\mathcal{R}}
\newcommand{\Cyclelength}{T_{\text{cyc}}}

\newcommand{\Arrivalrate}[1]{\lambda_{#1}}
\newcommand{\Queuelength}[1]{Q_{#1}}

\newcommand{\Lyap}[1]{f(#1)}

%% file: introduction.tex
\section{Introduction}
Traffic congestion is a significant issue in urban areas, leading to increased travel times, reduced productivity, and environmental concerns. A potential solution to this issue is Urban Air Mobility (UAM), which aims to use the urban airspace for on-demand mobility \cite{holden2016uber}. A crucial aspect of UAM systems, especially in high-demand regimes, is traffic management \cite{mueller2017enabling}. The objective of traffic management is to efficiently use the limited UAM resources, such as the airspace, takeoff and landing areas, and the aircraft, to meet the demand. The purpose of this paper is to systematically design and analyze a traffic management policy for on-demand UAM networks.
\par
The UAM traffic management problem can be considered as a natural extension of the classic Air Traffic Flow Management (ATFM) problem \cite{bertsimas1998air}. The objective of ATFM is to optimize the flow of commercial air traffic to ensure safe and efficient operations in the airspace system, considering factors such as airspace and airport capacity constraints, weather conditions, and operational constraints \cite{bertsimas1998air, bertsimas2011integer}. The first departure point in the context of UAM is the unpredictable nature of demand. Unlike commercial air traffic where the demand is highly predictable weeks in advance, the UAM systems will be designed to provide on-demand services. This poses a significant planning challenge. 
\par
To address this problem, recent works such as \cite{chin2023protocol} have attempted to incorporate fairness considerations into the existing ATFM formulation to accommodate the on-demand nature of UAM. Other solutions include heuristic approaches such as first-come first-served scheduling \cite{pradeep2018heuristic} and simulations \cite{bosson2018simulation}. While previous works provide valuable insights into the operation of UAM systems, they do not explicitly address two critical aspects. First is the concept of \emph{rebalancing}: the UAM aircraft will need to be constantly redistributed in the network when the demand for some destinations is higher than others. Efficient rebalancing ensures the effectiveness and sustainability of on-demand UAM systems. The concept of rebalancing has been explored extensively in the context of on-demand ground transportation \cite{pavone2012robotic}. However, these studies predominantly use flow-level formulations which do not capture the safety and separation considerations associated with aircraft operations. The second aspect which has not been addressed in the UAM literature is a thorough characterization of the system-level throughput. Roughly, the throughput of a given traffic management policy determines the highest demand that the policy can handle \cite{POOLADSANJ2023TRC}. In the context of UAM, the throughput is tightly related to the notion of passenger waiting time. In particular, the throughput determines the demand threshold at which the expected passenger waiting time transitions from being stabilized to being increasing over time. Therefore, it is desirable to design a policy that achieves the maximum possible throughput.
\par
In light of the aforementioned gaps in the literature, we present a centralized traffic management framework for on-demand UAM networks. We propose a scheduling policy, called VertiSync, which synchronously schedules the aircraft for either servicing trip requests or rebalancing in the network. The primary contributions of this paper are as follows:
\begin{enumerate}
    \item Developing a scheduling policy, called VertiSync, for on-demand UAM networks, subject to aircraft safety margins and energy requirements.
    \item Incorporating the aspect of rebalancing into the UAM scheduling framework.
    \item Characterizing the system-level throughput of VertiSync, and demonstrating its effectiveness through a case study for the city of Los Angeles.
\end{enumerate}
\par
The rest of the paper is organized as follows: in Section~\ref{sec:problem-formulation}, we describe the problem formulation. We provide our traffic management policy and characterize its throughput in Section~\ref{section:scheduling}. We provide the Los Angeles case study in Section~\ref{section:simulation}, and conclude the paper in Section~\ref{section:conclusion}.

%% file: problem-formulation.tex
\section{Problem Formulation}\label{sec:problem-formulation}

\input{network}

\input{separation}

\input{demand}

%% file: network.tex
\subsection{UAM Network Structure}
\mpcommentout{
A UAM network will consist of a number of take-off/landing areas, called \emph{vertiports}, connected by a set of routes. Each vertiport may have multiple takeoff/landing pads, called \emph{vertipads}.
}
We describe a UAM network by a directed graph $\mathcal{G}$. A node in the graph $\mathcal{G}$ represents either a \emph{vertiport}, i.e., take-off/landing area, or an intermediate point where two or more routes cross paths. A link in the graph $\mathcal{G}$ represents a section of the routes that have the link in common. We let $\vertset$ be the set of vertiports and $\vertcard{}$ be its cardinality. We let $N_v$ be the total number of \emph{vertipads}, i.e., takeoff/landing pads, at vertiport $v \in \vertset$. An Origin-Destination (O-D) pair $p$ is an ordered pair $p=(o_p,d_p)$ where $o_p, d_p \in \vertset$ and there is at least one route from $o_p$ to $d_p$. We let $\routeset$ be the set of O-D pairs and $\routecard$ be its cardinality; see Figure~\ref{fig:graph-example}. To simplify the network representation and without loss of generality, we assume that each vertiport has exactly one outgoing link exclusively used for takeoffs from that vertiport, and a separate incoming link exclusively used for landings. For simplicity (and lack of existing routes), we also assume that there is at most one route between any two vertiports and that the UAM routes do not conflict with the current airspace.
\begin{figure}[t]
    \centering
    \includegraphics[width=0.3\textwidth]{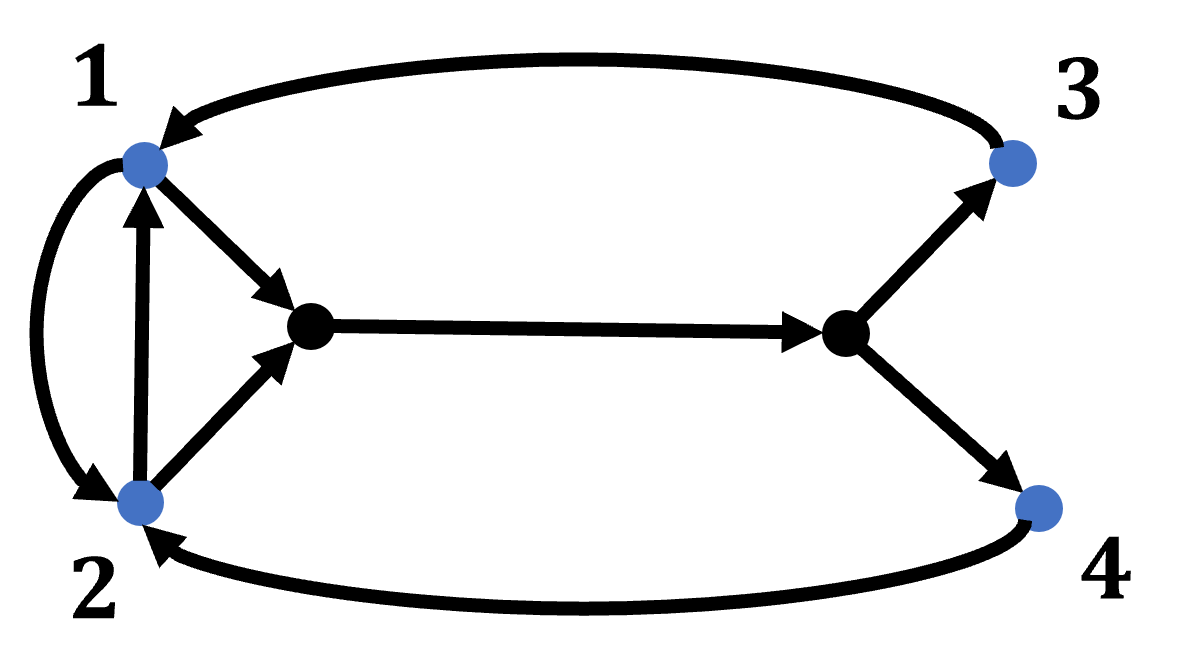}
    
    \vspace{0.1 cm}
    
    \caption{\sf A UAM network with $\vertcard{}=4$ vertiports (blue circles) and $\routecard=8$ O-D pairs $(1,3)$, $(1,4)$, $(2,3)$, $(2,4)$, $(3,1)$, $(4,2)$, $(1,2)$, and $(2,1)$.}
    \label{fig:graph-example}
\end{figure}
\begin{remark}
Given an O-D pair $p = (o_p, d_p)$, the opposite pair $q = (d_p, o_p)$ may or may not be an O-D pair. However, to enable rebalancing, it is natural to assume that there always exists a collection of routes that connect any two vertiports. 
\end{remark}
\par
In the next section, we will discuss the constraints associated with a UAM aircraft flight operation.

%% file: separation.tex
\subsection{Operational Constraints}\label{section:op-constraints}
In this section, we describe the constraints and assumptions related to UAM aircraft flight operations. Let $\aircraftset$ be the set of aircraft in the system and $\aircraftcard$ be its cardinality. Each aircraft's flight operation consists of the following three phases: 
\begin{itemize}
    \item \textbf{takeoff:} during this phase, the aircraft is positioned on a departure vertipad and passengers (if any) are boarded onto the aircraft before it is ready for takeoff. To position the aircraft on the departure vertipad, it is either transferred from a parking space or directly lands from a different vertiport. Let $\takeofftau$ denote the \emph{takeoff separation}, which represents the minimum time required between successive aircraft takeoffs from the same vertipad. In other words, the takeoff operations are completed in a $\takeofftau$-minute time window for every flight, which implies that the takeoff rate from each vertipad is at most one aircraft per $\takeofftau$ minutes. 
    \item \textbf{airborne:} to ensure safe operation, all UAM aircraft must maintain appropriate horizontal and vertical safety distance from each other while airborne. We assume that all UAM aircraft have the same characteristics so that these margins are the same for all the aircraft. Without loss of generality, we assume that different links of the graph $\mathcal{G}$ are at a safe horizontal and vertical distance from each other, except in the vicinity of the nodes where they intersect. Let $\cruisetau$ be the minimum time between two aircraft takeoffs with the same route from the same vertiport, ensuring that all the airborne safety margins are satisfied. Therefore, the takeoff rate from each vertiport is at most one aircraft per $\cruisetau$ minutes. We assume that $\takeofftau \geq \cruisetau$, i.e., the takeoff separation is more restrictive than the separation imposed by the safety margin, and $k_{\tau}:= \takeofftau/\cruisetau$ is integer-valued.  
    \item \textbf{landing:} once the aircraft lands, passengers (if any) are disembarked, new passengers (if any) are embarked, and the aircraft undergoes servicing. Thereafter, the aircraft is either transferred to a parking space or, if it has boarded new passengers or needs to be rebalanced, takes off to another vertiport. Similar to takeoff operations, we assume that the landing operations are completed within a $\takeofftau$-minute time window for every flight. That is, once an aircraft lands, the next takeoff or landing can occur after $\takeofftau$ minutes. Therefore, if two aircraft with the same route take off from the same vertipad at least $\takeofftau$ minutes apart, then they will be able to land on the same vertipad at their destination. We assume that the parking capacity at each vertiport is at least $\aircraftcard$ so that an arriving aircraft always clears the vertipad after landing. 
    \end{itemize}
\begin{remark}
The above assumptions regarding the takeoff and landing separations, as well as the airborne safety margins, are practical given the current technological limitations \cite{bosson2018simulation, vascik2017constraint}. However, our results can be generalized and are not limited to these specific assumptions.
\end{remark}
\par
In addition to the above assumptions, we consider an ideal case where there is no external disturbance such as adverse weather conditions. As a result, if an aircraft's flight trajectory satisfies the safety margins and the separation requirements, then the aircraft follows it without deviating from the trajectory. On the other hand, if its trajectory does not satisfy either of the safety or separation requirements, we assume that a lower-level controller, e.g., a pilot or a remote operator, handles the safe operation of the aircraft. We do not specify this controller in the paper since we only consider policies that guarantee before takeoff that the aircraft's route is clear and a vertipad is available for landing.

%% file: demand.tex
\subsection{Demand and Performance Metric}\label{section:demand}
In an on-demand UAM network, the demand is likely not known in advance. We use exogenous stochastic processes to capture the unpredictable nature of the demand. It will be convenient for performance analysis later on to adopt a discrete time setting. Let the duration of each time step be $\cruisetau$, which represents the the minimum time between two aircraft takeoffs with the same route from the same vertiport that guarantees the safety margins. The number of trip requests for an O-D pair $p \in \routeset$ is assumed to follow an i.i.d Bernoulli process with parameter $\Arrivalrate{p}$ independent of other O-D pairs. That is, at any given time step, the probability that a new trip is requested for the O-D pair $p$ is $\Arrivalrate{p}$ independent of everything else. Note that $\Arrivalrate{p}$ specifies the rate of new requests for the O-D pair $p$ in terms of the number of requests per $\cruisetau$ minutes. Let $\Arrivalrate{} := (\Arrivalrate{p})$ be the vector of arrival rates. 
\par
For each O-D pair, the trip requests are queued up in an unlimited capacity queue until they are \emph{serviced}, at which point they leave the queue. In order to be serviced, a request
must be assigned to an aircraft, and the aircraft must take
off from the verriport. A \emph{scheduling policy} is a rule that schedules the aircraft in the system for either servicing trip requests or rebalancing, i.e., taking off without passengers to service trip requests at other vertiports.
\par
The objective of the paper is to design a policy that can handle the maximum possible demand under the operational constraints discussed in Section~\ref{section:op-constraints}. The key performance metric to evaluate a policy is the notion of \emph{throughput} which we will now formalize. For $p \in \routeset$, let $\Queuelength{p}(t)$ be the number of trip requests in the queue for the O-D pair $p$ at time $t$. Let $\Queuelength{}(t) = (\Queuelength{p}(t))$ be the vector of trip requests for all the O-D pairs at time $t$. We define the \emph{under-saturation} region of a policy $\pi$ as 
\begin{equation*}
    U_{\pi} = \{\Arrivalrate{}: \limsup_{t \to \infty} \E{\Queuelength{p}(t)} < \infty~~\forall p \in \routeset~ \text{under policy}~\pi\}.
\end{equation*}
\par
This is the set of $\Arrivalrate{}$'s for which the expected number of trip requests remain bounded for all the O-D pairs. The boundary of this set is called the throughput of the policy $\pi$. We are interested in finding a policy $\pi$ such that $U_{\pi'} \subseteq U_{\pi}$ for all policies $\pi'$, including those that have information about the demand $\Arrivalrate{}$. In other words, if the network remains under-saturated using some policy $\pi'$, then it also remains under-saturated using the policy $\pi$. In that case, we say that policy $\pi$ maximizes the throughput for the UAM network. In the next section, we introduce one such policy.

%% file: scheduling.tex
\section{Network-Wide Scheduling}\label{section:scheduling}

\mpcommentout{
\begin{figure}[t]
    \centering
    \includegraphics[width=0.3\textwidth]{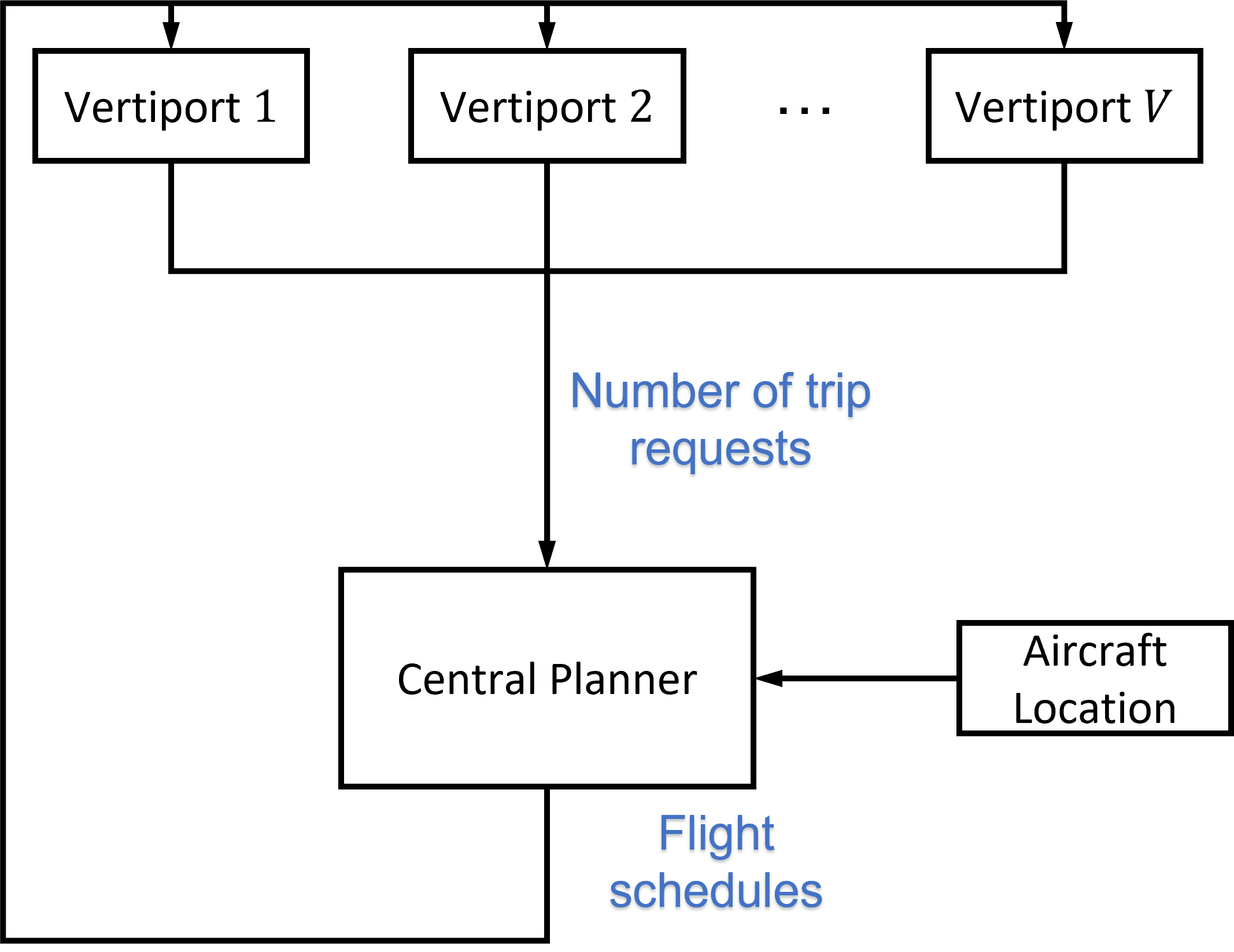}
    
    \vspace{0.1 cm}
    
    \caption{\sf The information flow in the VertiSync policy at the beginning of a cycle.}
    \label{fig:renewal-schematic}
\end{figure}
}
\subsection{VertiSync Policy}
We now introduce our policy which is inspired by the queuing theory literature \cite{armony2003queueing} and the classical Traffic Flow Management Problem (TFMP) formulation \cite{bertsimas1998air}. The policy works in cycles during which only the trips that were requested before the start of the cycle are serviced. At the start of a cycle, a central planner schedules the aircraft for either servicing trip requests or rebalancing in the network until all the trips that were requested before the start of the cycle are serviced. The aircraft schedule during a cycle is communicated to vertiport operators responsible for takeoff and landing operations at each vertiport. It is assumed that the central planner knows the location and energy state of each aircraft as well as the number of trip requests for each O-D pair. The scheduling for service or rebalancing is done synchronously during a cycle, and hence the name of the policy. 
\par
To conveniently track aircraft locations in discrete time, we introduce the notion of \emph{slot}. For each O-D pair $p \in \routeset$, a slot represents a specific position along the route of that pair, such that if two aircraft with the same route occupy adjacent slots, then they satisfy the airborne safety margins. In particular, consider an aircraft's flight trajectory that satisfies the safety margins and separation requirements. At the end of every $\cruisetau$ minutes along the aircraft's route, a slot is associated with the aircraft's position, with the first slot located at the origin vertiport. If an aircraft occupies the first slot of an O-D pair $p$, it means that it is positioned on a vertipad at the origin vertiport $o_p$. Similarly, if it occupies the last slot of the O-D pair $p$, it means that it has landed on a vertipad at the destination vertiport $d_p$. Consider a configuration of slots for all the O-D pairs established at time $t=0$. Without loss of generality, we assume that if a link is common to two or more routes, then the slots associated with those routes coincide with each other on that link. Additionally, if two aircraft with different routes occupy adjacent slots, then they will satisfy the airborne safety margins with respect to each other. We also let the first and last slots on each link coincide with the tail and head of that link, respectively. We assign a unique identifier to each slot, with overlapping slots having different identifiers. Let $\slotset{p}$ be the set of slots associated with the O-D pair $p$, and let $S_p$ be its cardinality. 
\par
Let $t$ be a fixed time, and $n=0,1,\ldots$. A key decision variable in the VertiSync policy is $w_{i,n}^{a,p}$, where $w_{i,n}^{a,p} = k$ if aircraft $a \in \aircraftset$ has visited slot $i \in \slotset{p}$ of the O-D pair $p$, $k$ times in the interval $(t,t+n\cruisetau]$. For brevity, the time $t$ is dropped from $w_{i,n}^{a,p}$ as it will be clear from the context. By definition, $w_{i,n}^{a,p}$ is non-decreasing with respect to $n$. Moreover, if $w_{i,n}^{a,p}-w_{i,n-1}^{a,p} = 1$ for some $n \geq 1$, then aircraft $a$ has occupied slot $i$ at some time in the interval $(t+(n-1)\cruisetau, t+n\cruisetau]$. We use the notation $w_{o,n}^{a,p}$ to represent the number of times aircraft $a$ with route $p$ has taken off from vertiport $o_p$ in the interval $(t,t+n\cruisetau]$. Similarly, $w_{d,n}^{a,p}$ indicates the number of times aircraft $a$ with route $p$ has landed on vertiport $d_p$ in the interval $(t,t+n\cruisetau]$. For slot $i \in \slotset{p}$, $i+1$ denotes its following slot, i.e., the slot that comes after slot $i$ along the route of the O-D pair $p$. Given two O-D pairs $p, q \in \routeset$ and slots $i \in \slotset{p}$ and $j \in \slotset{q}$, we let $i=j$ if slot $i$ coincides with slot $j$. Finally, we use the binary variable $u_{v,n}^{a}$ to denote whether aircraft $a$ can begin its takeoff phase of the flight operation from vertiport $v \in \vertset$ at time $t+n\cruisetau$ (if $u_{v,n}^{a}=1$) or not (if $u_{v,n}^{a}=0$).
\begin{definition}\label{def:renewal-policy}
\textbf{(VertiSync Policy)} 
The policy works in cycles of \emph{variable} length, with the first cycle starting at time $t_1 = 0$. At the beginning of the $k$-th cycle at time $t_k$, each vertiport communicates the number of trip requests originating from that vertiport to the central planner, i.e., the vector of trip requests $\Queuelength{}(t_k) = (\Queuelength{p}(t_k))$ is communicated to the central planner. During the $k$-th cycle, only theses requests will be serviced.
\par
The central planner solves the following optimization problem to determine the aircraft schedule while minimizing the total airborne time of all the aircraft. That is, the central planner aims to minimize
\begin{equation}\label{eq:TFMP-objective}
\sum_{n=1}^{M_k}\sum_{a \in \aircraftset}\sum_{p \in \routeset}(w_{o,n}^{a,p}-w_{o,n-1}^{a,p})T_p,
\end{equation}
where $M_k \in \N$ is such that $M_k \cruisetau$ is a conservative upper-bound on the $k$-th cycle length, $w_{o,n}^{a,p}$ is the number of times aircraft $a$ with route $p \in \routeset$ has taken off from vertiport $o_p$ in the interval $(t_k,t_k+n\cruisetau]$, and $T_p$ is the flight time from vertiport $o_p$ to $d_p$ when the aircraft satisfies the airborne safety margins and separation requirements. The following constraints must be satisfied:
    \begin{subequations}\label{eq:TFMP-constraints}
        \begin{align}
            &\sum_{a \in \aircraftset}w_{o,M_k}^{a,p} \geq \Queuelength{p}(t_k),\hspace{1.65cm} \forall p \in \routeset, \label{eq:TFMP-constraint-serviceall} \\
             &w_{i,n-1}^{a,p}-w_{i,n}^{a,p} \leq 0, \notag \\ 
             &\hspace{1.5cm}\forall n \in [M_k],~ p \in \routeset,~i \in \slotset{p},~ a \in \aircraftset, \label{eq:TFMP-constraint-w-increasing}\\
             &\sum_{p \in \routeset}\sum_{i \in \slotset{p}}(w_{i,n}^{a,p}-w_{i,n-1}^{a,p}) \leq 1,\hspace{0.4cm}\forall n \in [M_k],~a \in \aircraftset, \label{eq:TFMP-constraint-aircraft-one-slot} \\
             &w_{i,n-1}^{a,p}=w_{i+1,n}^{a,p}, \notag \\
             &\hspace{1.4cm}\forall n \in [M_k],~p \in \routeset, ~i \in \slotset{p}:i \neq d,~a \in \aircraftset, \label{eq:TFMP-constraint-aircraft-next-slot}\\
             &\sum_{a \in \aircraftset}(w_{i,n}^{a,p}-w_{i,n-1}^{a,p}) + (w_{j,n}^{a,q}-w_{j,n-1}^{a,q}) \leq 1, \notag \\
             &\hspace{1.4cm}~\forall n \in [M_k],~ p,q \in \routeset,~ i \in \slotset{p}, j \in \slotset{q}: i=j, \notag\\
             &\hspace{4.88cm} i,j \neq o,d, \label{eq:TFMP-constraint-air-safety}\\
             &\sum_{a \in \aircraftset}\sum_{p \in \routeset: o_p=v}(w_{o,n}^{a,p}-w_{o,n-k_{\tau}}^{a,p}) \leq N_v, \notag\\
             &\hspace{1.4cm}~\forall n \in \{k_{\tau},\ldots,M_k\},~v \in \vertset, \label{eq:TFMP-constraint-takeoff-separation} \\
             &w_{o,n}^{a,p}-w_{o,n-1}^{a,p} \leq u_{v,n-k_{\takeofftau}}^{a}, \notag \\
             &\hspace{1.4cm}~\forall n \in \{k_{\takeofftau},\ldots,M_k\},~p \in \routeset:o_p=v,~a \in \aircraftset, \label{eq:TFMP-constraint-aircraft-availability} \\
             &u_{v,n-k_{\takeofftau}+1}^{a} = u_{v,n-k_{\takeofftau}}^{a} - (w_{o,n}^{a,p}-w_{o,n-1}^{a,p}),\notag\\
             &\hspace{1.4cm}~\forall n \in \{k_{\takeofftau},\ldots,M_k\},~p \in \routeset:o_p=v,~a \in \aircraftset, \label{eq:TFMP-constraint-parking-takeoff} \\
             &\sum_{a \in \aircraftset}\left(\sum_{p \in \routeset: o_p=v}(w_{o,n}^{a,p}-w_{o,n-1}^{a,p}) \right. \notag \\
             &\hspace{2cm} \left. +\sum_{q \in \routeset: d_q=v}(w_{d,n-1}^{a,q}-w_{d,n-u_{\tau}}^{a,q})\right) \leq N_v,\notag\\
             &\hspace{1.4cm}~\forall n \in \{k_{\tau},\ldots,M_k\},~v \in \vertset, \label{eq:TFMP-constraint-landing-separation} \\
             &u_{v,n}^{a} = u_{v,n-1}^{a} + (w_{d,n}^{a,p}-w_{d,n-1}^{a,p}),\notag\\
             &\hspace{1.4cm}~\forall n \in [M_k],~p \in \routeset:d_p=v,~a \in \aircraftset, \label{eq:TFMP-constraint-parking-landing}\\
             &w_{i,n}^{a,p} \in \N_{0},~u_{v,n}^{a} \in \{0,1\} \notag \\
             &\hspace{1.4cm}~\forall n \in [M_k],~v \in \vertset,~p \in \routeset,~i \in \slotset{p},~a \in \aircraftset. \notag
        \end{align}
    \end{subequations}
\par    
Constraint \eqref{eq:TFMP-constraint-serviceall} ensures that all the trip requests are serviced by the end of the cycle. Constraint \eqref{eq:TFMP-constraint-w-increasing} enforces the decision variable $w_{i,n}^{a,p}$ to be non-decreasing in time. Constraint \eqref{eq:TFMP-constraint-aircraft-one-slot} ensures that each aircraft $a$ occupies at most one slot in the network at any time, and constraint \eqref{eq:TFMP-constraint-aircraft-next-slot} guarantees that if aircraft $a$ occupies slot $i$ at some time $t \in (t_k+(n-1)\cruisetau,t_k+n\cruisetau]$, then it will occupy slot $i+1$ at time $t+\cruisetau$. Constraint \eqref{eq:TFMP-constraint-air-safety} ensures the airborne safety margins by allowing at most one aircraft occupying any overlapping or non-overlapping slot at any time. Similarly, constraints \eqref{eq:TFMP-constraint-takeoff-separation} and \eqref{eq:TFMP-constraint-landing-separation} ensure that the takeoff and landing separations are satisfied at every vertiport, respectively. Constraint \eqref{eq:TFMP-constraint-aircraft-availability} enforces that aircraft $a$ can take off from a vertiport at time $t$ only if it has landed at that vertiport at or before time $t-\takeofftau$, and constraints \eqref{eq:TFMP-constraint-parking-takeoff} and \eqref{eq:TFMP-constraint-parking-landing} update $u_{v,n}^{a}$ once aircraft $a$ takes off from vertiport $v$ and lands on vertiport $v$, respectively. 
\par
While traversing route $p$, aircraft $a$ expends energy $E_p$, which is calculated as the sum of the energy required for takeoff, cruise, and landing. Let $E_{n}^{a}$ be the remaining energy of aircraft $a$ at time $n$, and let $z_{v, n}^{a}$ be a binary variable indicating whether aircraft $a$ is being re-charged at time $n$ at vertiport $v$ (if $z_{v, n}^{a} = 1$) or not (if $z_{v, n}^{a} = 0$). In addition to the constraints in \eqref{eq:TFMP-constraints}, we require
\begin{subequations}\label{eq:TFMP-energy-constraints}
\begin{align}
    E_{n}^{a} &= E_{n-1}^{a} - \sum_{p \in \routeset}(w_{o,n}^{a,p}-w_{o,n-1}^{a,p})E_p + \sum_{v \in \vertset} z_{v, n}^{a} E_{\text{inc}}, \notag \\
    &\hspace{1.4cm}~\forall n \in [M_k],~a \in \aircraftset, \label{eq:TFMP-constraint-energy-balance} \\
    E_{n}^{a} &\geq 0, \notag \\
    &\hspace{1.4cm}~\forall n \in [M_k],~a \in \aircraftset, \label{eq:TFMP-constraint-energy-lower} \\
    E_{n}^{a} &\leq E_{\text{max}}, \notag \\
    &\hspace{1.4cm}~\forall n \in [M_k],~a \in \aircraftset, \label{eq:TFMP-constraint-energy-upper} \\
    z_{v, n}^{a} &\leq u_{v,n-k_{\takeofftau}}^{a}, \notag \\
    &\hspace{1.4cm}~\forall n \in \{k_{\takeofftau},\ldots,M_k\},~v \in \vertset,~a \in \aircraftset, \label{eq:TFMP-constraint-energy-binary}
\end{align}
\end{subequations}
where $E_{\text{inc}}$ is the energy increment of an aircraft during one time step while being recharged, and $E_{\text{max}}$ is the maximum energy of an aircraft. Constraint \eqref{eq:TFMP-constraint-energy-balance} is the balance equation for the energy state of aircraft $a$, and constraints \eqref{eq:TFMP-constraint-energy-lower} and \eqref{eq:TFMP-constraint-energy-upper} limit the minimum and maximum energy of aircraft $a$. Constraint \eqref{eq:TFMP-constraint-energy-binary} ensures that aircraft $a$ can be recharged at a vertiport only if it is available at that vertiport, i.e., $u_{v,n-k_{\takeofftau}}^{a} = 1$. 
\par
The initial values $u_{v,0}^{a}$, $z_{v, 0}^{a}$, and $w_{i,0}^{a,p}$ are determined by the location and energy state of aircraft $a$ at the end of the previous cycle. For example, if aircraft $a$ has occupied slot $i$ of the O-D pair $p$ at the end of cycle $k-1$, then $u_{v,0}^{a}=0$ for all $v \in \vertset$, $w_{j,0}^{a,p} = 1$ for slot $j=i$ and any other slot $j \in \slotset{p}$ that precedes slot $i$, i.e., comes before slot $i$ along the route of the O-D pair $p$, and $w_{j,0}^{a,q} = 0$ for any other $q \neq p$ and $j \in \slotset{q}$. 
The $k$-th cycle ends once all the requests for that cycle have been serviced. 
\begin{remark}
Note that the VertiSync policy only requires real-time information about the number of trip requests, but does not require any information about the arrival rate. This makes VertiSync a suitable option for an actual UAM network where the arrival rate is unknown or could vary over time.
\end{remark}
\end{definition}
\subsection{VertiSync Throughput}
We next characterize the throughput of the VeriSync policy. To this end, we introduce a $\routecard$-dimensional \emph{service vector} $\routingvec{}{i}=(\routingvec{p}{i})$. 
If $\routingvec{}{i}$ is activated, then aircraft can continuously takeoff from the origin vertiport $o_p$ at the rate of $\routingvec{p}{i}$ per $\cruisetau$ minutes without violating the airborne safety margins and separation requirements. If $\routingvec{p}{i} = 0$, then the takeoff rate for the O-D pair $p$ is zero.
Let $\safeschset$ be the set of all non-zero service vectors, and $\numsafesch$ be its cardinality. We use $\routingvec{}{i} = (\routingvec{p}{i})$, $i \in [\numsafesch]$, to denote a particular vector in $\safeschset$. Note that each $\routingvec{}{i} \in \safeschset$ is associated with at least one schedule that, upon availability of aircraft, guarantees continuous takeoffs for the O-D pair $p$ at the rate of $\routingvec{p}{i}$ aircraft per $\cruisetau$ minutes without violating the safety margins and separation requirements. Recall the aircraft operational constraints from Section~\ref{section:op-constraints}, and note that $\routingvec{p}{i}$ is an integer multiple of $\cruisetau/\takeofftau$ and $0 \leq \routingvec{p}{i} \leq 1$.
\par
\begin{example}\label{ex:scheduling-vectors}
Consider the network in Figure~\ref{fig:graph-example}. We number the O-D pairs $(1,3)$, $(1,4)$, $(2,3)$, $(2,4)$, $(3,1)$, $(4,2)$, $(1,2)$, and $(2,1)$ as $1$ to $8$, respectively. Suppose that each vertiport has only one vertipad. Let the takeoff separation be $\takeofftau = 5$ minutes, and $\cruisetau = 0.5$ minutes. Due to symmetry, if an aircraft for the O-D pair $1$ takes off at $t=0$, then an aircraft for the O-D pair $4$ can take off at $t=\cruisetau$ minute without violating the airborne safety margins. Therefore, $\routingvec{}{1}=(0.1,0,0,0.1,0,0,0,0)$ is a service vector in $\safeschset$ with the takeoff schedules $t=0, 5, 10, \ldots$, and $t=0.5, 5.5, 10.5, \ldots$, for the O-D pairs $1$ and $4$, respectively. Similarly, $\routingvec{}{2}=(0,0.1,0.1,0,0,0,0,0)$ and $\routingvec{}{3}=(0,0.1,0,0,0,0,0,0)$ are two other service vectors in $\safeschset$.
\end{example}
By using the service vectors $\routingvec{}{i} \in \safeschset$, a feasible solution to the optimization problem \eqref{eq:TFMP-objective}-\eqref{eq:TFMP-energy-constraints} can be constructed as follows: (i) activate at most one service vector $\routingvec{}{i} \in \safeschset$ at any time, (ii) while $\routingvec{}{i}$ is active, schedule available aircraft to take off at the rate of $\routingvec{p}{i}$ per $\cruisetau$ minutes for any O-D pair $p \in \routeset$, (iii) switch to another service vector in $\safeschset$ provided that the safety margins and separation requirements are not violated after switching, and (iv) repeat (i)-(iii) until all the requests for the $k$-th cycle are serviced. 
\par
The next theorem provides an inner-estimate of the throughput of the VertiSync policy when the number of aircraft $\aircraftcard$ is sufficiently large and the following ``symmetry" assumption holds:
\begin{assumption}\label{assumption:symmetry}
    For any service vector $\routingvec{}{i} \in \safeschset$, there exists a service vector $\routingvec{}{j} \in \safeschset$ such that for all $p \in \routeset$ with $\routingvec{p}{i} > 0$, $\routingvec{p}{j} = \routingvec{p}{i}$ and $\routingvec{q}{j} = \routingvec{p}{i}$, where $q = (d_p,o_p)$ is the opposite O-D pair to the pair $p$. In words, by using the service vector $\routingvec{}{j}$, aircraft can continuously take off at the same rate for the O-D pair $p$ and its opposite pair $q$ without violating the safety margins and separation requirements. Let $\slotcard{}$ be the total number of slots, with overlapping slots being considered a single slot. 
\end{assumption}
\begin{theorem}\label{thm:renewal-sufficient}
If the UAM network satisfies the symmetry Assumption~\ref{assumption:symmetry}, and the number of aircraft satisfies
\begin{equation*}
    \aircraftcard \geq \slotcard{} + \sum_{p \in \routeset} \max_{i}\lceil\frac{E_{p}}{E_{\text{inc}}} \routingvec{p}{i} \rceil,
\end{equation*}
then the VertiSync policy can keep the network under-saturated for demands belonging to the set 
\begin{equation*}
    D^{\circ}=\{\Arrivalrate{}: \Arrivalrate{} < \sum_{i = 1}^{\numsafesch}\routingvec{}{i} x_i,~\text{\small for}~ x_i \geq 0,~ i \in [\numsafesch],~ \sum_{i = 1}^{\numsafesch}x_i \leq 1\},
\end{equation*}
where the vector inequality $\Arrivalrate{} < \sum_{i = 1}^{\numsafesch}\routingvec{}{i} x_i$ is considered component-wise.
\end{theorem}
\begin{proof}
See Appendix~\ref{section:proof-renewal-sufficient}.
\mpcommentout{
The proof can be found in the arXiv version of this paper. 
}
\end{proof}
\subsection{Fundamental Limit on Throughput}
In this section, we provide an outer-estimate on the throughput of any \emph{safe-by-construction} policy. A safe-by-construction policy is a policy that guarantees before takeoff that the aircraft's entire route will be clear and a veripad will be available for landing. Since the UAM aircraft have limited energy reserves, it is desirable to use safe-by-construction policies for traffic management purposes \cite{thipphavong2018urban}. 
\par
Any safe-by-construction policy uses the service vectors in $\safeschset$, either explicitly or implicitly, to schedule the aircraft.
\mpcommentout{
This is done by activating one or multiple service vectors from $\safeschset$, scheduling the aircraft to take off at the rates specified by the activated service vectors, and switching between service vectors provided that the safety margins and separation requirements are not violated after switching.
}
Although it is possible for a safe-by-construction policy to activate multiple service vectors at any time, we may restrict ourselves to policies that activate at most one service vector from $\safeschset$ at any time. This restriction does not affect the generality of safe-by-construction policies being considered; by activating at most one service vector at any time and rapidly switching between service vectors in $\safeschset$, it is possible to achieve an exact or arbitrarily close approximation of any safe schedule while ensuring the safety margins and separation requirements. 
\par
The next result provides a fundamental limit on the throughput of any safe-by-construction policy.
\begin{theorem}\label{thm:necessary}
If a safe-by-construction policy $\pi$ keeps the network under-saturated, then the demand must belong to the set 
\begin{equation*}
\begin{aligned}
    D=\{\Arrivalrate{}: \Arrivalrate{} \leq \sum_{i = 1}^{\numsafesch}\routingvec{}{i} x_i,~\text{\small for}~ x_i \geq 0, ~i \in [\numsafesch],~ \sum_{i = 1}^{\numsafesch}x_i \leq 1\},
\end{aligned}
\end{equation*}
where the vector inequality $\Arrivalrate{} \leq \sum_{i = 1}^{\numsafesch}\routingvec{}{i} x_i$ is considered component-wise.
\end{theorem}
\begin{proof}
See Appendix \ref{section:proof-necessary}.
\end{proof}

%% file: simulation.tex
\section{Simulation Results}\label{section:simulation}
\begin{figure}[t]
    \centering
    \includegraphics[width=0.35\textwidth]{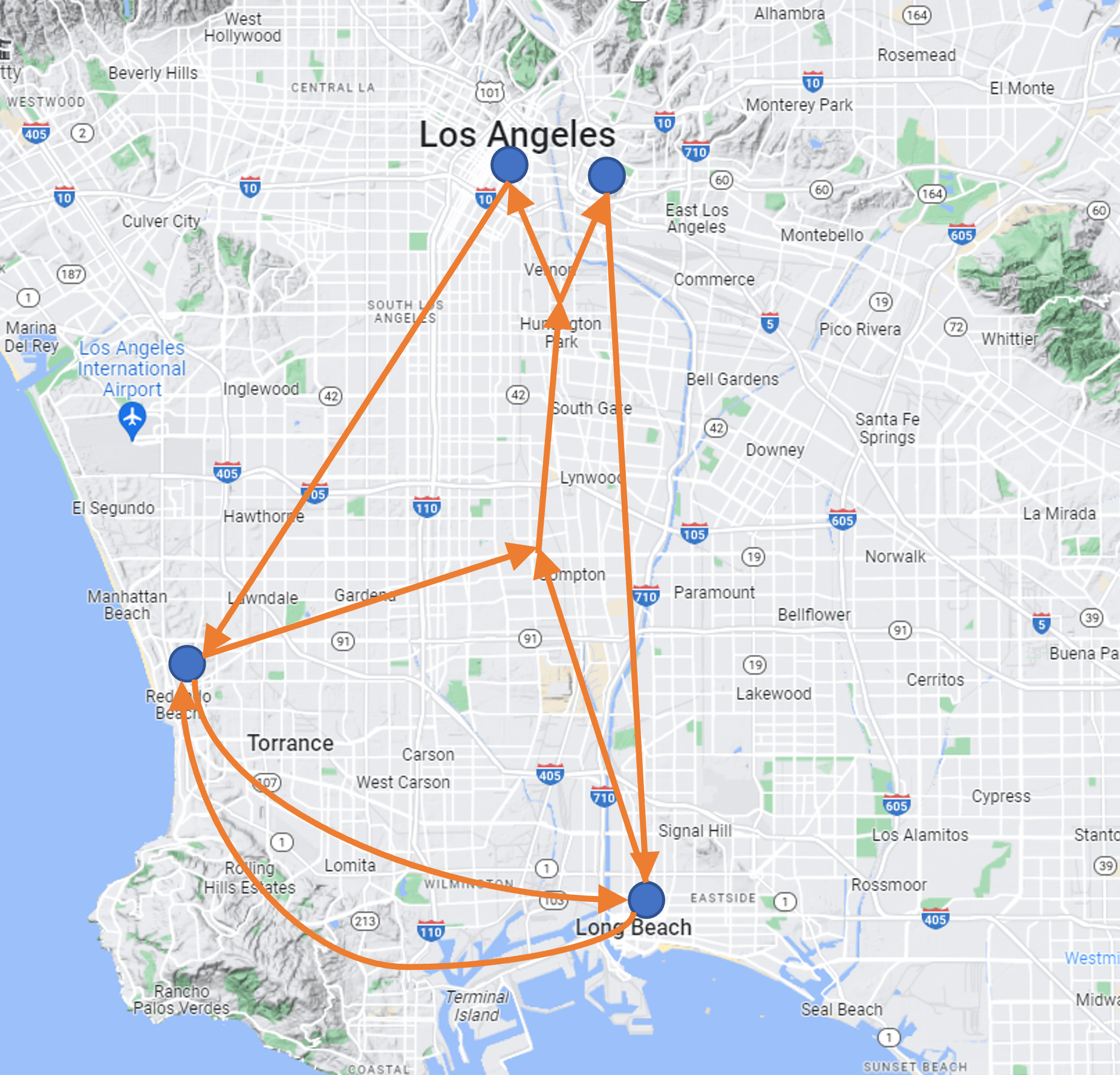}
    
    \vspace{0.1 cm}
    
    \caption{\sf The UAM network for the city of Los Angeles. The blue circles show the vertiports and the orange arrows show the links.}
    \label{fig:LA-case-study}
\end{figure}
\begin{figure}[t]
    \centering
    \includegraphics[width=0.35\textwidth]{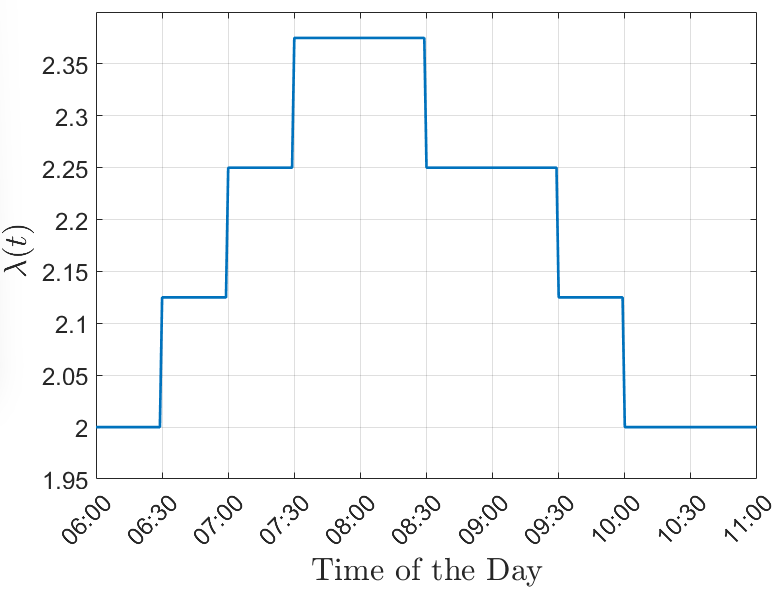}
    
    \vspace{0.1 cm}
    
    \caption{\sf The rate of trip requests per $\takeofftau$ minutes ($\Arrivalrate{}(t)$).}
    \label{fig:arrival-rate}
\end{figure}
In this section, we demonstrate the performance of the VertiSync policy and compare it with a heuristic traffic management policy from the literature. As a case study, we select the city of Los Angeles, which is anticipated to be an early adopter market due to the severe road congestion, existing infrastructure, and mild weather \cite{vascik2017constraint}. All the simulations were performed in MATLAB.
\par
We consider four vertiports located in Redondo Beach (vertiport $1$), Long Beach (vertiport $2$), and the Downtown Los Angeles area (vertiports $3$ and $4$). The choice of vertiport locations is adopted from \cite{vascik2017constraint}. Each vertiport is assumed to have $10$ vertipads. Figure~\ref{fig:LA-case-study} shows the network structure, where there are $8$ O-D pairs. We let the takeoff and landing separations $\takeofftau$ be $5$~[min] and let $\cruisetau = 0.5$~[min]. We let the flight time for the O-D pairs $(1,2)$ and $(2,1)$ be $5$~[min], and for the rest of the O-D pairs be $8$~[min]. We simulate this network during the morning period from 6:00-AM to 11:00-AM, during which the majority of demand originates from vertiports $1$ and $2$ to vertiports $3$ and $4$. We let the trip requests for each of the O-D pairs $(1,3)$, $(1,4)$, $(2,3)$, and $(2,4)$ follow a Poisson process with a piece-wise constant rate $\Arrivalrate{}(t)$. The demand for other O-D pairs is set to zero during the morning period. With a slight abuse of notation, we scale $\Arrivalrate{}(t)$ to represent the number of trip requests per $\takeofftau$ minutes. From Theorem~\ref{thm:necessary}, given $\Arrivalrate{}(t) = \Arrivalrate{}$, the necessary condition for the network to remain under-saturated is that $\Arrivalrate{} \leq \takeofftau/4\cruisetau = 2.5$ trip requests per $\takeofftau$ minutes, i.e., $\rho := 4\Arrivalrate{}\cruisetau/\takeofftau \leq 1$. Figure~\ref{fig:arrival-rate} shows $\Arrivalrate{}(t)$, where we have considered a heavy demand between 7:00-AM to 9:30-AM to model the morning rush hour, i.e., $\rho(t) = 4\Arrivalrate{}(t)\cruisetau/\takeofftau \in [0.9,1)$ between 7:00-AM to 9:30-AM. 
\par
\begin{figure}[t]
    \centering
    \includegraphics[width=0.35\textwidth]{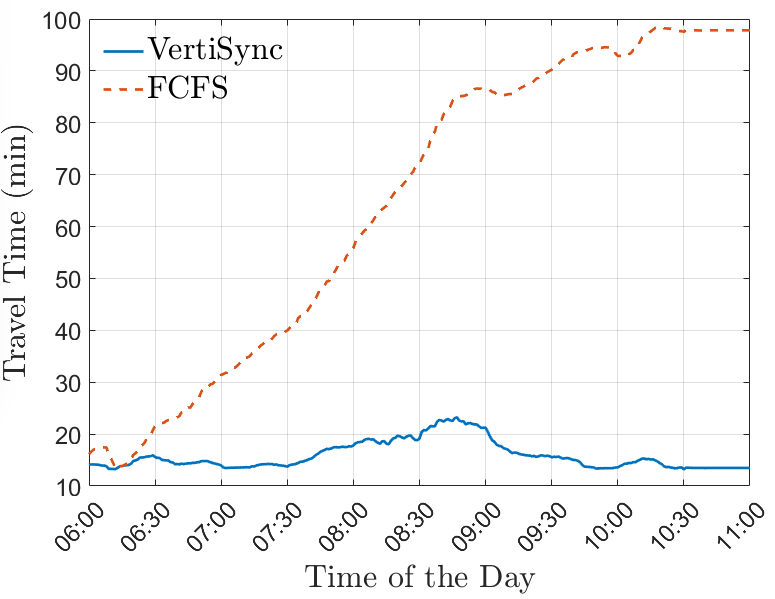}
    
    \vspace{0.1 cm}
    
    \caption{\sf The travel time under the VertiSync and FCFS policies for the demand $\Arrivalrate{}(t)$.}
    \label{fig:travel-time-comparison}
\end{figure}
\begin{figure}[t]
    \centering
    \includegraphics[width=0.35\textwidth]{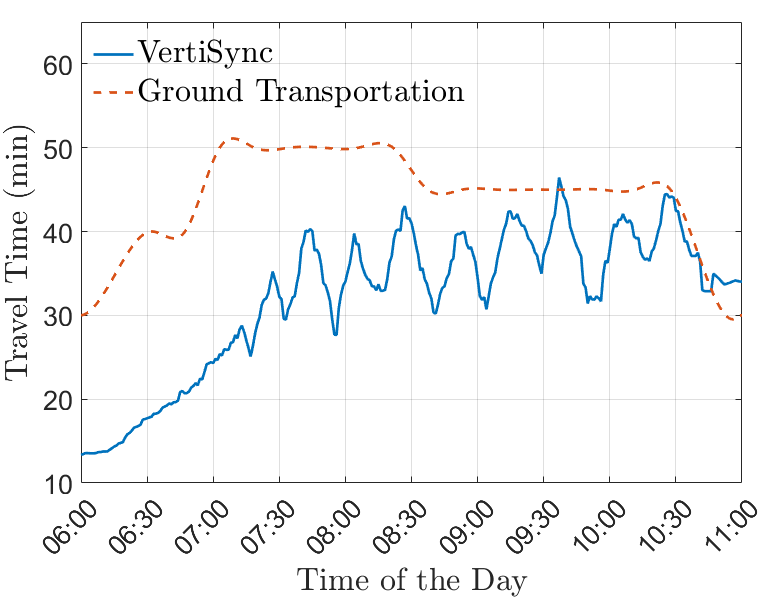}
    
    \vspace{0.1 cm}
    
    \caption{\sf The travel time under the VertiSync policy when the demand is increased to $1.2\Arrivalrate{}(t)$ (over-saturated regime), and the ground transportation travel time.}
    \label{fig:travel-time-ground}
\end{figure}
We first evaluate the travel time under our policy and the \emph{First-Come First-Serve} (FCFS) policy \cite{pradeep2018heuristic}. The FCFS policy is a heuristic policy which schedules the trip requests in the order of their arrival at the earliest time that does not violate the safety margins and separation requirements. We let the number of aircraft be $\aircraftcard = 32$, and assume that all of them are initially located at vertiport $1$. We also assume that an aircraft is always available to service a trip request at its scheduled time under the FCFS policy. Finally, the optimization problem \eqref{eq:TFMP-objective}-\eqref{eq:TFMP-constraints} in the VertiSync policy is solved analytically. This approach is made possible by the simple symmetrical network structure considered in the simulations. 
\par
For the above demand and a random simulation seed, $518$ trips are requested during the morning period from which the FCFS policy services $411$ before 11:00~AM while the VertiSync policy is able to service all of them. Figure~\ref{fig:travel-time-comparison} shows the (passenger) travel time, which is computed by averaging the travel time of all trips requested within each $10$-minute time interval. The travel time of a trip is computed from the moment that trip is requested until it is completed, i.e., reached its destination. As expected, the VertiSync policy keeps the network under-saturated since $\Arrivalrate{}(t) < 2.5$ for all $t$. However, the FCFS policy fails to keep the network under-saturated due to its greedy use of the vertipads and UAM airspace which is inefficient. 
\par
We next evaluate the demand threshold at which the VertiSync policy becomes less efficient than ground transportation. Figure~\ref{fig:travel-time-ground} shows the travel time under the VertiSync policy when the demand is increased to $1.2\Arrivalrate{}(t)$. By Theorem~\ref{thm:necessary}, the network is in the over-saturated regime from 6:30-AM to 10:00-AM since $1.2\Arrivalrate{}(t) > 2.5$ trip requests per $\takeofftau$ minutes. However, as shown in Figure~\ref{fig:travel-time-ground}, the travel time is still less than the ground travel time during the morning period. The ground travel times are collected using the Google Maps service from 6:00-AM to 11:00-AM on Thursday, May 19, 2023 from Long Beach to Downtown Los Angeles (The travel times from Redondo Beach to Downtown Los Angeles were similar). 
\par
\mpcommentout{
\begin{figure}[t]
    \centering
    \includegraphics[width=0.35\textwidth]{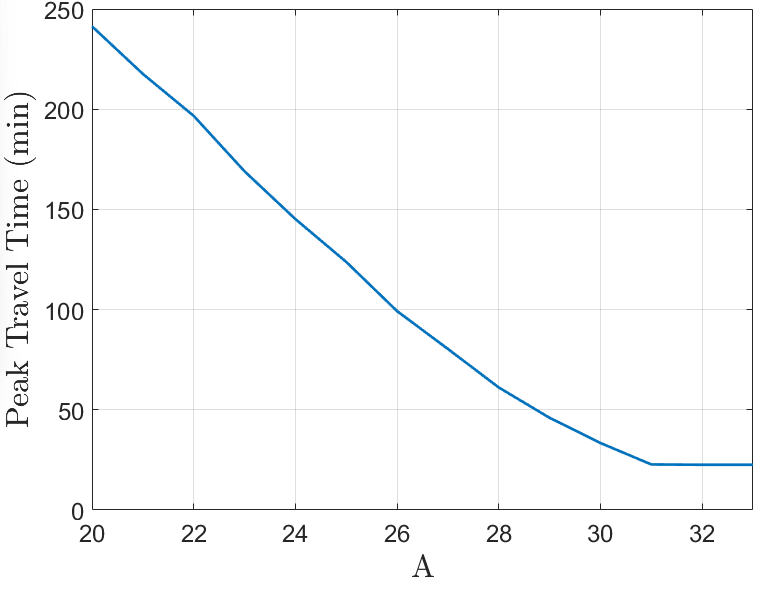}
    
    \vspace{0.1 cm}
    
    \caption{\sf The effect of the number of aircraft $\aircraftcard$ on peak travel time under the VertiSync policy and demand $\Arrivalrate{}(t)$.}
    \label{fig:travel-time-fleetsize}
\end{figure}
Finally, we evaluate the effect of the number of aircraft $\aircraftcard$ on the peak travel time under the VertiSync policy. Figure~\ref{fig:travel-time-fleetsize} shows the peak travel time versus the number of aircraft, where the peak travel time is averaged over $1000$ simulation rounds with different seeds. For the demand $\Arrivalrate{}(t)$, the necessary number of aircraft to keep the network under-saturated is $\aircraftcard = 31$. Moreover, the travel time does not improve as $\aircraftcard$ exceeds $32$. 
}

%% file: conclusion.tex
\section{Conclusion}\label{section:conclusion}
In this paper, we provided a traffic management policy for on-demand UAM networks and analyzed its throughput. We conducted a case study for the city of Los Angeles and showed that our policy significantly improves travel time compared to a first-come first-serve policy. We plan to expand our case study to more complex networks with more origin and destination pairs and study the computational requirements of the optimization problem in our policy. We also plan to implement our policy in a high-fidelity air traffic simulator.

%% file: appendix.tex
\appendices

\input{proof-renewal-sufficient}

\input{proof-necessary}

%% file: proof-renewal-sufficient.tex
\section{Proof of Theorem \ref{thm:renewal-sufficient}}\label{section:proof-renewal-sufficient}
Consider the $(k+1)$-th cycle. We first construct a feasible solution to the optimization problem \eqref{eq:TFMP-objective}-\eqref{eq:TFMP-energy-constraints} by using the service vectors in $\safeschset$. Consider the Linear Program (LP) 
\begin{equation}\label{eq:linear-program}
    \begin{aligned}
        \text{Minimize}&~\sum_{i=1}^{\numsafesch}K_i \\
        \text{Subject to}&~ \sum_{i=1}^{\numsafesch}\routingvec{}{i}K_i \geq \Queuelength{}(t_{k+1}), \\
        &~ K_i \geq 0,~ i \in [\numsafesch],
    \end{aligned}
\end{equation}
where the inequality $\sum_{i=1}^{\numsafesch}\routingvec{}{i}K_i \geq \Queuelength{}(t_{k+1})$ is considered component-wise. Let $K_{i}^{*}$, $i \in [\numsafesch]$, be the solution to \eqref{eq:linear-program}. A feasible solution to the optimization problem \eqref{eq:TFMP-objective}-\eqref{eq:TFMP-energy-constraints} can be constructed as follows: 
\begin{enumerate}
    \item Choose a service vector $\routingvec{}{i} \in \safeschset$ with $K_{i}^{*} > 0$. Without loss of generality, we may assume that $\routingvec{}{i}$ is such that for all $p \in \routeset$ with $\routingvec{p}{i} > 0$, $\routingvec{p}{j} = \routingvec{p}{i}$ and $\routingvec{q}{j} = \routingvec{p}{i}$, where $q = (d_p,o_p)$ is the opposite O-D pair to the pair $p$. Before activating $\routingvec{}{i}$, distribute the aircraft in the system so that for any $p \in \routeset$ with $\routingvec{p}{i} > 0$, there is
    \begin{equation*}
       \frac{\slotcard{p}}{k_{\takeofftau}-k_{\takeofftau}\routingvec{p}{i}+1} + \lceil\frac{E_{p}}{E_{\text{inc}}}\routingvec{p}{i}\rceil 
    \end{equation*}
    aircraft at vertiport $o_p$. The initial distribution of aircraft takes at most $\aircraftcard \overline{T}$ minutes, where $\overline{T} = \max_{p \in \routeset}T_p$. 
    \item Once the initial distribution is completed and the airspace is empty, activate $\routingvec{}{i}$ for a duration of $(K_{i}^{*}+1)\cruisetau$ minutes. During this time, aircraft with route $p$ take off (upon availability) at the rate $\routingvec{p}{i}$ from vertiport $o_p$. Hence, at most $\slotcard{p}/(k_{\takeofftau}-k_{\takeofftau}\routingvec{p}{i}+1)$ slots of route $p$ will be occupied by an aircraft during this time. In addition, an aircraft needs to have at least $E_p$ energy to traverse route $p$. Since the recharging takes $E_p/E_{\text{inc}}$ time steps, if $\lceil \routingvec{p}{i} E_p/E_{\text{inc}} \rceil$ aircraft with energy of at least $E_p$ are available at vertiport $o_p$, then we can ensure continuous takeoffs at the rate $\routingvec{p}{i}$. From step 1, the number of aircraft at vertiport $o_p$ is $\slotcard{p}/(k_{\takeofftau}-k_{\takeofftau}\routingvec{p}{i}+1) + \lceil \routingvec{p}{i} E_p/E_{\text{inc}} \rceil$. Similarly, there is also enough aircraft at veriport $d_p$, and from the assumption $\routingvec{q}{j} = \routingvec{p}{i}$, they can simultaneously take off at the rate $\routingvec{p}{i}$. Therefore, we can ensure continuous takeoffs at the rate $\routingvec{p}{i}$ for route $p$, which implies that, at the end of this step, $\routingvec{p}{i}K_{i}^{*}$ requests will be serviced for the O-D pair $p$.  
    \item Once step 2 is completed and the airspace is empty, repeat steps 1 and 2 for another vector in $\safeschset$. The amount of time it takes for the airspace to become empty at the end of step 2 is at most $\overline{T}$ minutes. Once each service vector $\routingvec{}{i} \in \safeschset$ with $K_{i}^{*} > 0$ have been activated, $\sum_{i=1}^{\numsafesch}\routingvec{p}{i}K_{i}^{*}$ requests will be serviced for each O-D pair $p \in \routeset$. From the constraint of the LP \eqref{eq:linear-program}, $\sum_{i=1}^{\numsafesch}\routingvec{}{i}K_{i}^{*} \geq \Queuelength{}(t_{k+1})$, i.e., all the requests for the $(k+1)$-th cycle will be serviced and the cycle ends. 
\end{enumerate}
By combining the time each of the above steps takes, it follows that
\begin{equation}\label{eq:cycle-upper-bound}
    \Cyclelength(k+1) \leq \sum_{i=1}^{\numsafesch}K_{i}^{*} + \frac{\numsafesch}{\cruisetau}(\aircraftcard \overline{T} + \overline{T} + \cruisetau),
\end{equation}
where $\Cyclelength(k+1) = (t_{k+2} - t_{k+1})/\cruisetau$.
\par
Without loss of generality, we assume that the ordering by which $\routingvec{}{i}$'s are chosen at each cycle are fixed and, when a cycle ends, the next cycle starts once the airspace becomes empty and the start time of is a multiple of $\cruisetau$. Finally, we assume that the initial distribution of aircraft before each $\routingvec{}{i}$ is activated takes $\aircraftcard \overline{T}$ minutes. With these assumptions, we can cast the network as a discrete-time Markov chain with the state $\{\Queuelength{}(t_k)\}_{k \geq 1}$. Since the state $\Queuelength{}(t_k) = 0$ is reachable from all other states, and $\mathbb{P}\left(\Queuelength{}(t_{k+1}) = 0~|~ \Queuelength{}(t_k) = 0\right) > 0$, the chain is irreducible and aperiodic. Consider the function $f: \Z_{+}^{\routecard} \ra [0,\infty)$
\begin{equation*}
    \Lyap{\Queuelength{}(t_k)} = \Cyclelength^{2}(k),
\end{equation*}
where $\Z_{+}^{\routecard}$ is the set of $\routecard$-tuples of non-negative integers. Note that $\Cyclelength(k)$ is a non-negative integer from our earlier assumption that the cycle start times are a multiple of $\cruisetau$. We let $\Lyap{\Queuelength{}(t_k)} \equiv \Lyap{t_k}$ for brevity.
\par

We start by showing that
\begin{equation}\label{eq:cycle-decreases-expectation}
    \limsup_{n \ra \infty}\E{\left(\frac{\Cyclelength(k+1)}{\Cyclelength(k)}\right)^2~\middle|~ \Cyclelength(k)=n} < 1.
\end{equation}
\par
To show \eqref{eq:cycle-decreases-expectation}, let $\Cyclelength(k)=n$, and let $\overline{A_p}(t_k,t_{k+1})$ be the cumulative number of trip requests for the O-D pair $p \in \routeset$ during the time interval $[t_k,t_{k+1})$. Note that $\Queuelength{p}(t_{k+1})=\overline{A_p}(t_k,t_{k+1})$, which implies from the strong law of large numbers that, with probability one, 
\begin{equation*}
    \lim_{n \ra \infty}\frac{\Queuelength{p}(t_{k+1})}{n} = \Arrivalrate{p}.
\end{equation*}
\par
By the assumption of the theorem, $\Arrivalrate{} \in D^{\circ}$. Hence, with probability one, there exists $N' > 0$ such that for all $n > N'$ we have $\Queuelength{}(t_{k+1})/n \in D^{\circ}$. Since $D^{\circ}$ is an open set, for a given $n > N'$, there exists non-negative $x_1, x_2, \ldots, x_{\numsafesch}$ with $\sum_{i=1}^{\numsafesch}x_i < 1$ such that $\Queuelength{}(t_{k+1})/n < \sum_{i=1}^{\numsafesch}\routingvec{}{i}x_i$. For $i \in [\numsafesch]$, define $K_i := nx_i$. Then, $K_i$'s are a feasible solution to the LP \eqref{eq:linear-program}, and $\sum_{i=1}^{\numsafesch}K_i < n$. Therefore, from \eqref{eq:cycle-upper-bound} and with probability one, it follows for all $n > N'$ that
\begin{equation*}
    \begin{aligned}
        \Cyclelength(k+1) &\leq \sum_{i=1}^{\numsafesch}K_{i}^{*} + \frac{\numsafesch}{\cruisetau}(\aircraftcard \overline{T} + \overline{T} + \cruisetau) \\
        &< n + \frac{\numsafesch}{\cruisetau}(\aircraftcard \overline{T} + \overline{T} + \cruisetau),
    \end{aligned}
\end{equation*}
which in turn implies, with probability one, that
\begin{equation}\label{eq:cycle-decreases-w.p.1}
    \limsup_{n \ra \infty}\left(\frac{\Cyclelength(k+1)}{n}\right)^2 < 1.
\end{equation}
\par
Finally, since the number of trip requests for each O-D pair is at most $1$ per $\cruisetau$ minutes, the sequence $\{\Cyclelength(k+1)/n\}_{n=1}^{\infty}$ is upper bounded by an integrable function. Hence, from \eqref{eq:cycle-decreases-w.p.1} and the Fatou's Lemma \eqref{eq:cycle-decreases-expectation} follows.
\par
We will now use \eqref{eq:cycle-decreases-expectation} to show that the network is under-saturated. Note that \eqref{eq:cycle-decreases-expectation} implies that there exists $\delta \in (0,1)$ and $N$ such that for all $n > N$ we have
\begin{equation*}
    \E{\left(\frac{\Cyclelength(k+1)}{\Cyclelength(k)}\right)^2 \middle|~ \Cyclelength(k) = n} < 1 - \delta,
\end{equation*}
which in turn implies that
\begin{equation*}
    \E{\Cyclelength^2(k+1) - \Cyclelength^2(k) \middle|~ \Cyclelength(k) > N} < -\delta \Cyclelength^2(k).
\end{equation*}
\par
Furthermore, $\Queuelength{p}(t_k) \leq \Cyclelength(k) \leq \Cyclelength^2(k)$ for all $p \in \routeset$, where the first inequality follows from the fact that $t_{k+1}-t_{k} \geq \Queuelength{p}(t_k)\cruisetau$ for any O-D pair $p \in \routeset$. Therefore, $\E{\Cyclelength^2(k+1) - \Cyclelength^2(k) \middle|~ \Cyclelength(k) > N} < -\delta\|\Queuelength{}(t_k)\|_{\infty}$, where $\|\Queuelength{}\|_{\infty} = \max_{p}\Queuelength{p}$. Finally, if $\Cyclelength(k) \leq N$, then $\Cyclelength(k+1) \leq 2 N \routecard (\overline{T}/\cruisetau + k_{\takeofftau}) =: b$. Therefore, 
\begin{equation*}
\begin{aligned}
    \E{\Cyclelength^2(k+1) \middle|~ \Cyclelength(k) \leq N} &\leq b^2 \\
    & +\Cyclelength^2(k)-\delta\|\Queuelength{}(t_k)\|_{\infty},
\end{aligned}
\end{equation*}
where we have used $\delta\|\Queuelength{}(t_k)\|_{\infty} \leq \|\Queuelength{}(t_k)\|_{\infty} \leq \Cyclelength^2(k)$. Combining all the previous steps gives
\begin{equation*}
    \E{f(t_{k+1}) - f(t_k) \middle|~ \Queuelength{}(t_k)} \leq -\delta \|\Queuelength{}(t_k)\|_{\infty} + b^2\mathds{1}_{B},
\end{equation*}
where $B = \{\Queuelength{}(t_k): f(t_k) \leq N^2\}$ (a finite set). From this and the well-known Foster-Lyapunov drift criterion \cite[Theorem 14.0.1]{meyn2012markov}, it follows that $\limsup_{t \to \infty} \E{\Queuelength{p}(t)} < \infty$ for all $p \in \routeset$, i.e., the network is under-saturated.

%% file: proof-necessary.tex
\section{Proof of Theorem \ref{thm:necessary}}\label{section:proof-necessary}
We use a proof by contradiction. Suppose that some safe-by-construction policy $\pi$ keeps the network under-saturated but $\Arrivalrate{} \notin D$. Then, for any non-negative $x_1, x_2, \ldots, x_{\numsafesch}$ with $\sum_{i =1}^{\numsafesch}x_i \leq 1$, there exists some O-D pair $p \in \routeset$ such that $\Arrivalrate{p} > \sum_{i = 1}^{\numsafesch}\routingvec{p}{i} x_i$.
\par
Without loss of generality, we may assume that whenever the service vector $\routingvec{}{i}$ becomes active, it remains active for a time interval that is a multiple of $\takeofftau$. Given $k \in \N_{0}$, let $t_k := k\takeofftau$, and let $x_i(t_k)$ be the proportion of time that the service vector $\routingvec{}{i}$ has been active under the policy $\pi$ up to time $t_k$. Then, $x_i := \limsup_{k \ra \infty}x_i(t_k) \geq 0$ for all $i \in [\numsafesch]$ and $\sum_{i = 1}^{\numsafesch}x_i \leq 1$. Therefore, there exists $p \in \routeset$ such that $\Arrivalrate{p} > \sum_{i = 1}^{\numsafesch}\routingvec{p}{i} x_i$. Note that when the service vector $\routingvec{}{i}$ is active, the flight requests for the O-D pair $p$ are serviced at the rate of at most $\routingvec{p}{i}/\takeofftau$. Hence, the number flight requests for the O-D pair $p$ that have been serviced by $\routingvec{}{i}$ up to time $t_k$ is at most $\routingvec{p}{i}x_i(t_k)t_k/\takeofftau$ up to time $t_k$. Let $\overline{A_p}(t_k) \equiv \overline{A_p}(0,t_k)$ be the cumulative number of flight requests for the O-D pair $p$ up to time $t_k$. We have
\begin{equation*}
    \Queuelength{p}(t_k) \geq \Queuelength{p}(0) + \overline{A_p}(t_k) - \sum_{i = 1}^{\numsafesch}\frac{\routingvec{p}{i} x_i(t_k) t_k}{\takeofftau},
\end{equation*}
which implies
\begin{equation*}
    \frac{\Queuelength{p}(t_k)}{t_k} \geq \frac{\Queuelength{p}(0)}{t_k} + \frac{\overline{A_p}(t_k)}{t_k} - \frac{1}{\takeofftau}\sum_{i = 1}^{\numsafesch}\routingvec{p}{i} x_i(t_k).
\end{equation*}
\par
By letting $k \ra \infty$, it follows from the strong law of large numbers that, with probability one,
\begin{equation*}
    \liminf_{k \ra \infty} \frac{\Queuelength{p}(t_k)}{k} \geq \Arrivalrate{p} - \sum_{i = 1}^{\numsafesch}\routingvec{p}{i} x_i.
\end{equation*}
\par
Since $\Arrivalrate{p} > \sum_{i = 1}^{\numsafesch}\routingvec{p}{i}x_i$, then, with probability one, $\liminf_{k \ra \infty} \Queuelength{p}(t_k)/k$ is bounded away from zero. Hence, 
\begin{equation*}
  \liminf_{k \ra \infty} \Queuelength{p}(t_k) = \infty.  
\end{equation*}
\par
Combining this with Fatou's Lemma imply that the expected number of flight requests for the O-D pair $p$ grows unbounded. This contradicts the network being under-saturated.